\documentclass[11pt]{amsart}
\usepackage[margin=1in]{geometry}
\calclayout
\usepackage[utf8]{inputenc}
\usepackage{amsmath, amsthm, amssymb, mathtools, enumitem, breqn, footnote, xcolor, tikz, appendix, pgfplots}
\usepackage[foot]{amsaddr}
\usepackage[symbol]{footmisc}
\usepackage[colorlinks=true, citecolor=blue]{hyperref}

\usetikzlibrary{patterns, decorations.shapes,  decorations.markings, calc}

\newtheorem{theorem}{Theorem}[section]
\newtheorem{lemma}{Lemma}[section]
\newtheorem{cor}{Corollary}[section]

\newtheorem{define}{Definition}[section]

\newcommand{\norm}[1]{\lVert #1 \rVert}
\newcommand{\norml}[1]{\lVert #1 \rVert_{\infty}}

\tikzstyle{obs} = [fill=teal!20, draw=teal!30, line width=1pt, rounded corners, rounded corners=0.5pt]
\tikzstyle{edge} = [fill=gray!40!white, draw=gray!60!black, dotted, line width=1pt, rounded corners=0.5pt]
\tikzstyle{mpath} = [draw=orange!80!black, dashed, line width=0.75pt, rounded corners=0.5pt]
\tikzstyle{robot} = [fill=orange!20, draw=orange!30, line width=1pt, rounded corners, rounded corners=0.5pt]
\tikzstyle{region} = [fill=red!20, draw=red!30]

\title{A Sublinear Algorithm for Path Feasibility Among Rectangular Obstacles}

\author{Alex Fan}
\email{htfan@mit.edu}

\author{Alicia Li}
\email{aliciali@mit.edu}

\author{Arul Kolla}
\address[Alex Fan, Alicia Li, Arul Kolla, Jason Gonzalez]{MIT}
\email{arulk@mit.edu}

\author{Jason Gonzalez}
\email{jgon@mit.edu}

\begin{document}

\begin{abstract}
The problem of finding a path between two points while avoiding obstacles is critical in robotic path planning. We focus on the feasibility problem: determining whether such a path exists. We model the robot as a query-specific rectangular object capable of moving parallel to its sides. The obstacles are axis-aligned, rectangular, and may overlap. Most previous works only consider nondisjoint rectangular objects and point-sized or statically sized robots. Our approach introduces a novel technique leveraging generalized Gabriel graphs and constructs a data structure to facilitate online queries regarding path feasibility with varying robot sizes in sublinear time. To efficiently handle feasibility queries, we propose an online algorithm utilizing sweep line to construct a generalized Gabriel graph under the $L_\infty$ norm, capturing key gap constraints between obstacles. We utilize a persistent disjoint-set union data structure to efficiently determine feasibility queries in $\mathcal{O}(\log n)$ time and $\mathcal{O}(n)$ total space.
\end{abstract}

\maketitle

\section{Introduction}

Finding the shortest path between two points that avoids colliding with a set of obstacles is an essential problem in robotic path planning. The variant we choose to study is the \textit{feasibility problem}: determining if there exists a collision-free path between two points. Specifically, we model the robot as a rectangle, which can move along the two axes. The size of the robot is specified in each query. Thus, our problem is the $L_1$ (rectilinear) path feasibility problem. Moreover, the obstacles are rectangular and axis-aligned with the robot. Also, the obstacles are not necessarily disjoint from each other, as our method can be used in the setting where the obstacles are rectilinear, meaning the objects are composed of all right angles and all their sides are parallel to a grid axes. We provide an online algorithm with a preprocessing step to construct the data structure and a search mechanism to handle queries, as we are tasked with determining path feasibility between two query points along with a query robot size at each time step. Our goal is to construct a data structure to facilitate these queries in sublinear time. 

The $L_1$ shortest collision-free path problem is generally approached from two angles: applying a continuous Dijkstra scheme \cite{maheshwari2018rectilinearshortestpathstransient}, or constructing a visibility graph and then performing search on it \cite{Chen2011}. The first paradigm typically involves propagating a wave front to facilitate the location of intersection points of potential path with obstacles. The second typically involves triangulating the free space, or constructing trees or track graphs. However, most previous works only consider a point robot or a static sized rectangular robot, whereas our method enables queries with different rectangular and circular robot sizes.

Our method first uses a sweep line algorithm to construct a planar graph which captures \textit{relevant} gaps between two obstacles. Intuitively, these gaps can be seen as a constraint on the maximum size of the robot which can pass between the obstacles to enter a different region. This graph is a generalized Gabriel graph under the $L_\infty$ norm. Then, we construct the dual of this graph, so now every vertex is a region bounded by relevant gaps and each edge represents a gap. These edges are weighted by the width of the corresponding gap. We can then use the persistent Disjoint-Set Union data structure. We add edges from the dual planar graph in order from smallest width to largest width. Using the DSU, we can quickly determine the feasibility of a path in $\mathcal{O}(\log n)$, as explained in more detail in Section 4.

\begin{figure}[h!]
    \centering
    \includegraphics[width=0.5\linewidth]{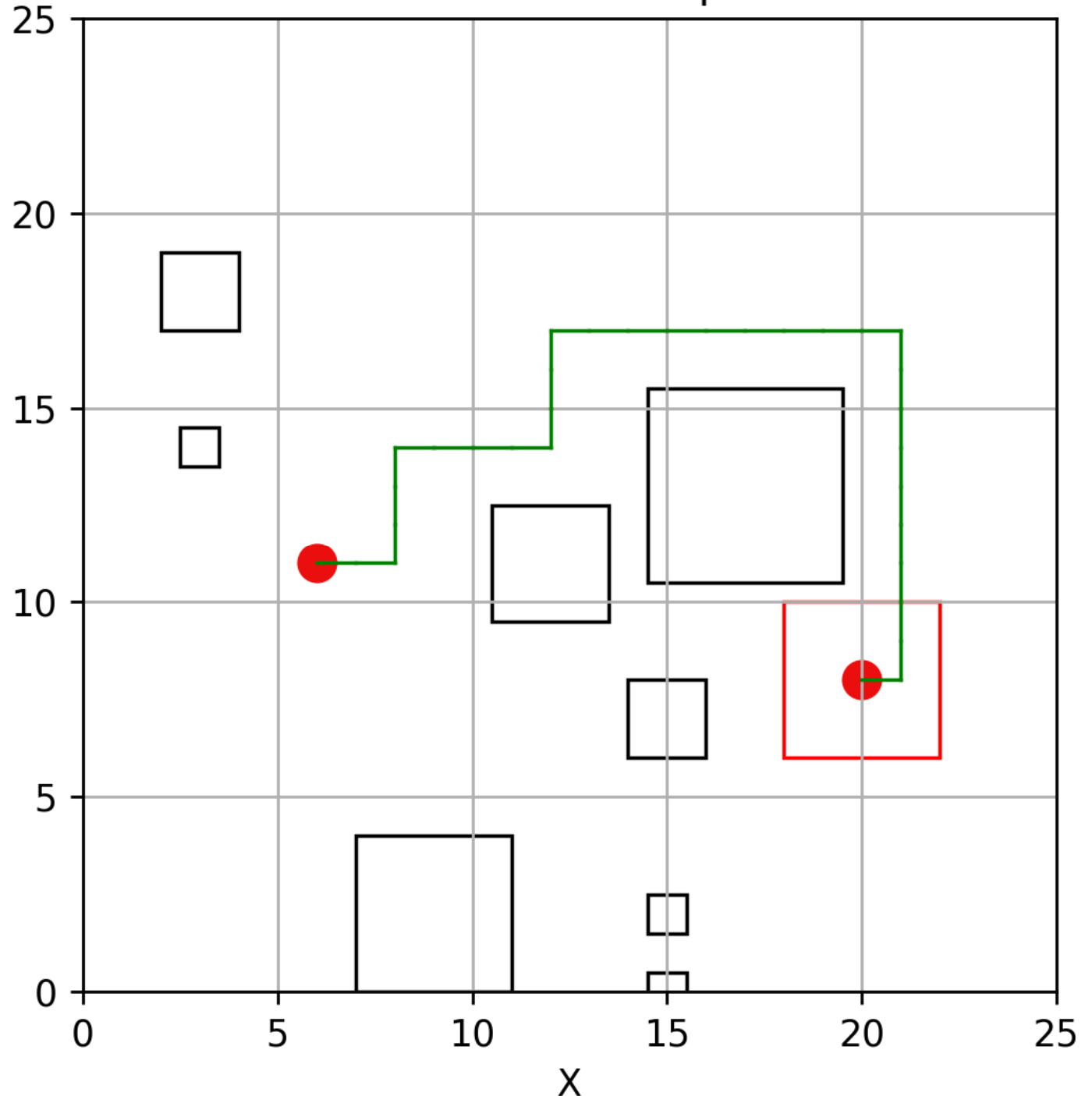}
    \caption{An example where the red dots are the starting and ending positions, and the green line is our trajectory, indicating that there is a feasible path.}
    \label{fig:problem_introduction}
\end{figure}

\section{Related Work}

While we are interested in the feasibility, not the optimality of $L_1$ paths, most previous works studying $L_1$ paths with polygonal obstacles studied shortest paths [\cite{rectilinearworlds},  \cite{Chen2011}, \cite{chen2012computingl1shortestpaths}]. For example, \cite{chen2012computingl1shortestpaths} and \cite{LL} consider the general case of polygonal obstacles. \cite{chen2012computingl1shortestpaths} supports queries of any two points, while \cite{LL} finds the minimum distances of a predetermined set of origin-destination nodes. \cite{chen2012computingl1shortestpaths} computes a shortest path map of size $\mathcal{O}(n)$ by triangulating the free space. They are able to answer queries in $\mathcal{O}(\log n)$ time, similar to our approach. However, their method assumes that the obstacles are disjoint. The disjoint assumption fundamentally changes the problem since most previous works assume that given any two obstacles, even if adjacent, there exists a path of width $0$ that a robot of width $0$ can pass through. This yields a strong guarantee that the optimal path is monotonic along one of the dimensions. Hence, obstacle setups like a spiral maze would not be possible, unlike in our work. Meanwhile, \cite{LL} constructs trees from each origin node where initial edges connect nodes which are directly accessible from each other without interfering objects. Then, they run a variant of Dijkstra's algorithm on these trees. 

Other previous works study shortest $L_1$ paths with rectangular objects. For instance, \cite{rectilinearworlds} considers the general case of rectangular boxes in higher dimensions, achieving a space complexity of $\mathcal{O}(n\log n)$ and a query time complexity of $\mathcal{O}(\log n)$ in two dimensions. This differs from our work because our method supports queries with different robot sizes. Meanwhile, \cite{10.1145/323233.323260} uses a graph sweep algorithm which leverages the fact that optimal paths are monotonic in at least one direction when the rectangular objects are disjoint. Their algorithm takes $\mathcal{O}(n\log n)$ for pre-processing and $\mathcal{O}(\log n)$ for each query. 

Most previous methods assume that the obstacles are disjoint and that the robot is a point. It is not obvious how to extend these works to the more realistic rectangular robots we consider. In particular, because the obstacles are assumed to be disjoint, the technique developed in \cite{10.1145/359156.359164}, in which the obstacle space is taken to be the Minkowski sum of the robot and the obstacles cannot be directly applied. To the extent of our knowledge, our work is the first to provide an efficient algorithm for path feasibility with a rectangular robot and obstacles which are not necessarily disjoint. Many important domains include non-disjoint rectangle obstacles, such as mazes and series of multiple rooms, for which the monotonicity assumption breaks down. Our work addresses this gap.

Another subfield of literature considers a point robot representation and objects that are rectilinear [\cite{rectilinearworlds}, \cite{1676904}, \cite{maheshwari2018rectilinearshortestpathstransient}]. For instance, \cite{1676904} constructs a so-called \textit{track graph} using a procedure similar to Dijkstra's algorithm in order to solve the one-to-all shortest paths problems on $n$ points, $e$ total edges, $k$ intersections of obstacle tracks, and $t$ extreme edges in $\mathcal{O}(n \min {\log n, \log e} +
(e + k) \log t$ time complexity. However, \cite{1676904} assumes the rectilinear obstacles are convex, whereas our method is applicable for any rectilinear obstacles. Meanwhile, \cite{maheshwari2018rectilinearshortestpathstransient} uses a continuous Dijkstra paradigm to find the shortest time path among transient rectilinear obstacles, meaning the obstacles may appear and disappear over time. The authors propagate a wave front, specifically a rhombus, from the source point and then conducts search over the intersections of the wave and obstacles to achieve an $\mathcal{O}(\log n)$ query time. This setting of point robot representations and rectilinear objects is equivalent to the setting of a rectangular robot and rectangular objects, since we can use the method described in \cite{10.1145/359156.359164}. However, many of these algorithms, including \cite{1676904}, only consider fixed target point and fixed robot size queries. Our algorithm supports queries of \textit{\textbf{different}} robot sizes, which cannot be easily done with rectilinear object methods in $\mathcal{O}(\log n)$ time.


\section{Preliminaries}

\subsection{Problem statement}
 We consider the problem of motion planning for a rectangular robot in a two-dimensional Euclidean space $\mathbb{R}^2$ with rectangular obstacles and establish the fundamental definitions and notation used throughout the paper. For brevity, we will analyze one instance of the problem for a fixed query.

Our first observation is that if our robot is rectangular, we can translate it into a square by scaling the $x$-dimension appropriately. Hence we define the problem for square robots below, which can easily be generalized to rectangular robots.

Define a set of $n$ \textit{axis-aligned} rectangular obstacles $\mathcal{O} = \{O_1, \dots, O_n\}$. Each obstacle $O_i$ is characterized by its lower-left corner $(x_1^i, y_1^i) \in \mathbb{R}^2$ and upper-right corner $(x_2^i, y_2^i) \in \mathbb{R}^2$, where $x_1^i < x_2^i$ and $y_1^i < y_2^i$. Formally, obstacle $O_i$ occupies the region $[x_1^i, x_2^i] \times [y_1^i, y_2^i]$. We also define the \textit{width} and \textit{height} of the obstacle in the usual way: $w_i = x_2^i - x_1^i$ and $h_i = y_2^i - y_1^i$.

The robot $R$ is defined as an axis-aligned square with side length $d \in \mathbb{R}^+$.


At any time $t$, the configuration of the robot can be represented by its center point $(x(t), y(t)) \in \mathbb{R}^2$. The robot only moves parallel to the axes. Given a configuration $(x, y)$, the robot occupies the region $(x - \frac{d}{2}, x + \frac{d}{2}) \times (y - \frac{d}{2}, y + \frac{d}{2})$.

A trajectory $\tau : [0, 1] \rightarrow \mathbb{R}^2$ is defined as a continuous function mapping time to robot configurations. For any trajectory to be feasible, it must satisfy the \textit{collision-free constraint}: at no point in time can the robot's occupied region intersect with any obstacle. Formally, for all $t \in [0, 1]$ and for all $O_i \in \mathcal{O}$, we require $R(\tau(t)) \cap O_i = \varnothing$, where $R(\tau(t))$ denotes the region occupied by the robot at configuration $\tau(t)$.

The motion planning problem can now be formally stated. Given a start configuration $s = (s_x, s_y)$ and goal configuration $t = (e_x, e_y)$, we seek to determine whether there exists a feasible trajectory $\tau$ such that $\tau(0) = s$ and $\tau(1) = t$. We assume that both start and goal configurations are collision-free, i.e., $R(s) \cap O_i = \varnothing$ and $R(t) \cap O_i = \varnothing$ for all $O_i \in \mathcal{O}$.

\subsection{Edges}

Given two square obstacles $O_i$ and $O_j$ from the obstacle set $\mathcal{O}$, we define a \textit{thin edge} $E_{i,j}$ between them when they are ``close enough'' to create a meaningful passage. Namely, given obstacles $O_i$ and $O_j$, a thin edge $E_{i,j}$ is defined as the rectangular region given in equation \ref{eqn:thin}.
\begin{equation}
    \label{eqn:thin}
    E_{i,j} := [\max(x_1^i, x_1^j), \min(x_2^i, x_2^j)] \times [\max(y_1^i, y_1^j), \min(y_2^i, y_2^j)].
\end{equation}

Visually, this can be seen as three distinct cases, shown in Figure \ref{fig:edging}.

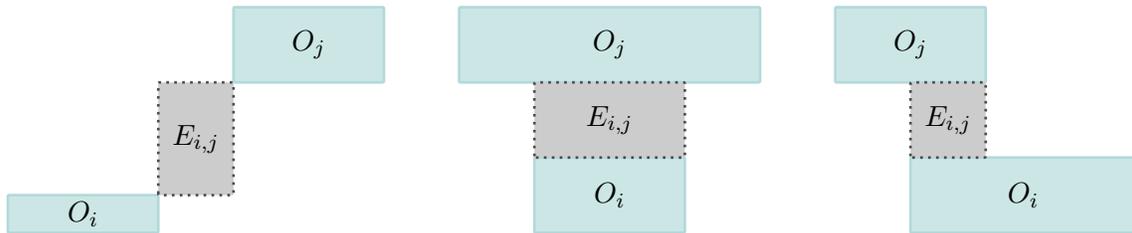
\begin{figure}[h!]
    \centering
    \begin{tikzpicture}[yscale=0.5]
    \draw[obs] (0,1) rectangle (2,2) node[midway] {$O_i$};
    \draw[obs] (3,5) rectangle (5,7) node[midway] {$O_j$};
    \draw[edge] (2,2) rectangle (3,5) node[midway] {$E_{i,j}$};

    \draw[obs] (7,1) rectangle (9,3) node[midway] {$O_i$};
    \draw[obs] (6,5) rectangle (10,7) node[midway] {$O_j$};
    \draw[edge] (7,3) rectangle (9,5) node[midway] {$E_{i,j}$};
    
    \draw[obs] (12,1) rectangle (15,3) node[midway] {$O_i$};
    \draw[obs] (13,5) rectangle (11,7) node[midway] {$O_j$};
    \draw[edge] (12,3) rectangle (13,5) node[midway] {$E_{i,j}$};
    \end{tikzpicture}
    \caption{Three cases for how a thin edge between obstacles is constructed.}
    \label{fig:edging}
\end{figure}

The main idea of the algorithm will be to show that by choosing a subset of edges to include, these edges and the obstacles will bound several regions in the plane where the robot has no obstacles.

\subsection{The Gabriel graph} 
\label{def:gabriel_graph}
The Gabriel graph (visualized in Figure \ref{fig:gabriel} for random points) is a fundamental structure in computational geometry that helps identify meaningful adjacencies between points. Given a set of points $P$ in the plane, two points $p_i, p_j \in P$ are connected by an edge in the Gabriel graph if and only if the circle with diameter $\overline{p_ip_j}$ (i.e., with $p_i$ and $p_j$ as antipodal points) contains no other points from $P$ in its interior.

\begin{figure}[h!]
    \centering
    \includegraphics[scale=0.125]{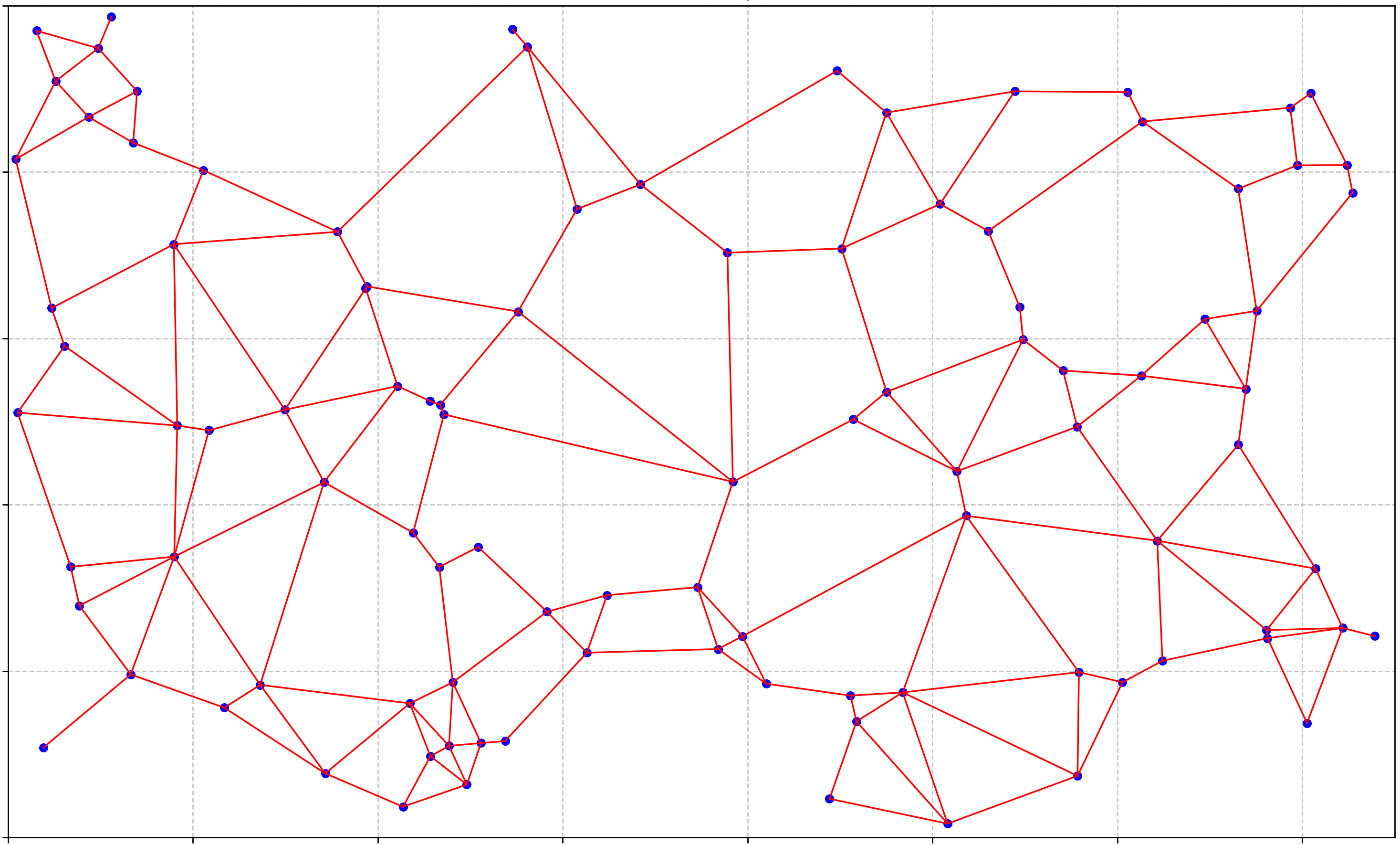}
    \caption{An example of a Gabriel graph on random points.}
    \label{fig:gabriel}
\end{figure}

In our context, a generalized Gabriel graph provides a way to determine which pairs of obstacles should be considered ``adjacent'' for the purposes of region partitioning. Specifically, we use the centers of obstacles as the point set for constructing our generalized Gabriel graph. The resulting edges identify obstacles that form relevant constraints for robot motion, as formalized below.

\begin{figure}[h!]
    \centering
    \begin{tikzpicture}[yscale=0.5]
    \draw[obs] (0,1) rectangle (2,2) node[midway] {$O_i$};
    \draw[obs] (3,5) rectangle (5,7) node[midway] {$O_j$};
    \draw[edge] (2,2) rectangle (3,5) node[midway] {$E_{i,j}$};
    \draw[mpath] (2,0) rectangle (3,7);

    \draw[obs] (7,1) rectangle (9,3) node[midway] {$O_i$};
    \draw[obs] (6,5) rectangle (9.5,7) node[midway] {$O_j$};
    \draw[edge] (7,3) rectangle (9,5) node[midway] {$E_{i,j}$};
    \draw[mpath] (6,3) rectangle (10,5);
    
    \draw[obs] (12,1) rectangle (15,3) node[midway] {$O_i$};
    \draw[obs] (13,5) rectangle (11.4,7) node[midway] {$O_j$};
    \draw[edge] (12,3) rectangle (13,5) node[midway] {$E_{i,j}$};
    \draw[mpath] (11,3) rectangle (14,5);
    \end{tikzpicture}
    \caption{Three cases of minimum pathways for obstacles, shown by the dashed lines.}

    \label{fig:min_pathway}
\end{figure}
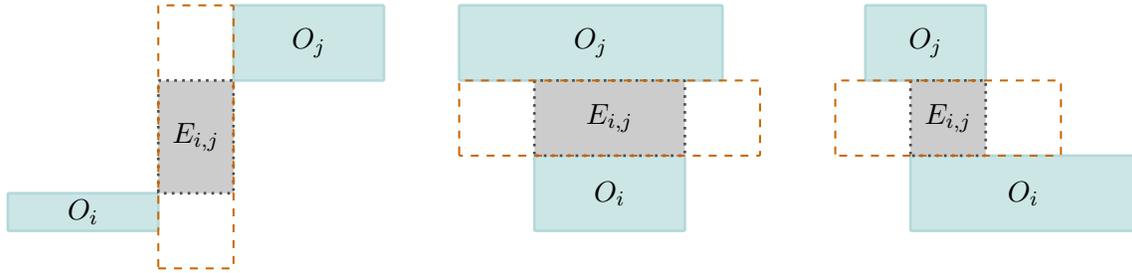

\begin{define}
    \label{ref_gabriel_graph}
    For each pair of obstacles $O_i$ and $O_j$, we define their \textbf{minimum pathway} as the smallest area that needs to be clear for a robot of maximum size to fully utilize the constraint. If there are no other obstacles in the minimum pathway, then the constraint is said to be \textbf{relevant}. Formally, there are three types of pathways depending on whether the two obstacles are completely overlapping, partially overlapping, or diagonal, which are illustrated in Figure \ref{fig:min_pathway}.
\end{define}

An intuitive definition of the minimum pathways is to imagine a robot with the maximum allowed width to squeeze in from either end, pass through the pathway, and exit from the other end.

Formally, for each pair of obstacles that form a relevant constraint, we will construct a thin edge between them. This will divide the plane into several \textit{regions}, a notion that is formalized below.

\subsection{Region partitioning} We formalize the intuition mentioned above.
\begin{define}
Given a set of obstacles $\mathcal{O}$ and a collection of thin edges $\mathcal{E}$, we partition the plane into distinct \textbf{regions}, which are the empty parts of the plane bounded by obstacles or thin edges. A visualization of regions is shown in Figure \ref{fig:region}.
\end{define}

\begin{figure}[h!]
    \centering
    \begin{tikzpicture}[xscale=1.5]
    \draw[region] (-1,0.5) rectangle (0,1.75) node[midway] {$R_1$};
    \draw[region] (1,0.5) rectangle (2,2) node[midway] {$R_2$};
    
    \draw[obs] (0,1) rectangle (1,2.5) node[midway] {$O_1$};
    \draw[obs] (2,3) rectangle (4,2) node[midway] {$O_2$};
    \draw[obs] (-1,0.5) rectangle (2.5,-0.5) node[midway] {$O_3$};
    \draw[obs] (-2,1.75) rectangle (-1.5,2.75) node[midway] {$O_4$};

    \draw[edge] (1,2) rectangle (2,2.5) node[midway] {$E_{1,2}$};
    \draw[edge] (2,0.5) rectangle (2.5,2) node[midway] {$E_{2,3}$};
    \draw[edge] (0,0.5) rectangle (1,1) node[midway] {$E_{1,3}$};
    \draw[edge] (-1.5,1.75) rectangle (0,2.5) node[midway] {$E_{1,4}$};
    \draw[edge] (-1.5,0.5) rectangle (-1,1.75) node[midway] {$E_{3,4}$};
    \end{tikzpicture}
    \caption{An example of several obstacles, thin edges, and two of the regions they form.}
    \label{fig:region}  
\end{figure}
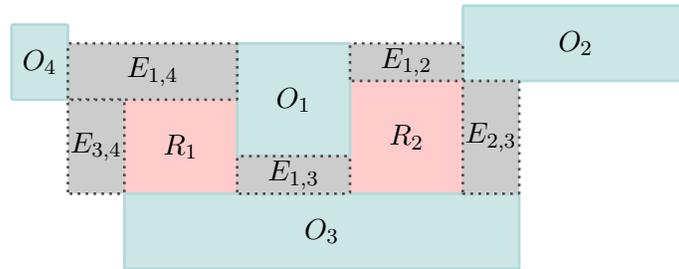

We can represent this partition as a planar graph $G = (V, E)$, where each vertex $v \in V$ corresponds to a region of the partition, and each thin edge $e \in E$ represents a shared boundary between adjacent regions. Namely, each edge $e \in E$ is weighted by the ``bottleneck width'' of the corresponding boundary~--- that is, the minimum distance between obstacles along that boundary. 


\subsection{Constraint equivalency}
\label{subsection_equivalence}

The main idea in our approach is that the continuous collision-free constraint for all points along a trajectory can be decomposed into two simpler types of constraints: pairwise passability constraints and region constraints. The pairwise passability constraints ensure that the robot can move between adjacent regions through their shared boundary, while the region constraints ensure that the robot can freely move within each region. Using it, we can reduce the infinite-dimensional problem of checking all possible robot positions to a finite set of constraints that can be verified efficiently.

\begin{theorem}[Constraint Equivalence]
\label{thm:constraint}
A feasible trajectory exists between two configurations $s$ and $t$ if and only if there exists a sequence of regions $R_1, R_2, ..., R_k$ from our region partitioning such that
\begin{enumerate}
    \item $s \in R_1$ and $g \in R_k$, and
    \item for each consecutive pair of regions $(R_i, R_{i+1})$, there exists a collision-free path between any point in $R_i$ to any point in $R_{i+1}$ (the so-called pairwise passability).
\end{enumerate}
\end{theorem}

The formal proof of this theorem is a little involved, so we relegate it to Appendix \ref{app:constr}. However, the rough idea is relatively straightforward; we can move ``freely'' in a region by continuity (a region is bounded by obstacles and edges), and by pairwise passability we can move between regions.

This decomposition principle is particularly powerful because it allows us to transform the continuous motion planning problem into a discrete graph search problem on the dual graph $G^*$, where edge weights represent bottleneck widths and the existence of a path in the graph corresponds directly to the feasibility of the trajectory in continuous space. Thus, the correctness of our approach is guaranteed by the fact that any feasible trajectory must respect these regional and passability constraints, and, conversely, any solution satisfying these constraints can be used to construct a feasible trajectory.


\section{Algorithms}
We now introduce our algorithm in the following fashion:
\begin{enumerate}
    \item We first explore a simpler variant of our problem to introduce relevant ideas and some intuitions.
    \item We prove that, by selecting certain edges, we can obtain a planar graph that can serve as a valid region partitioning scheme.
    \item We provide an algorithm that constructs a valid partition in $\mathcal{O}(n \log n)$ time, where $n$ is the number of points that define our obstacles.
    \item Given a valid partition and its dual planar graph, we provide an algorithm that can answer online queries in $\mathcal{O}(\log n)$ time.
\end{enumerate}

\subsection{Circular obstacles and robots}
Before attacking the main problem, it is helpful to outline an argument for the alternate case where the robot is a circle and all obstacles are congruent circles. 

The key insight is that each pair of obstacles induces a constraint on the maximum radius of robots that can pass between them. However, many of these constraints are redundant~--- specifically, if a third obstacle lies within the circle whose diameter is the line segment between two obstacles, then the constraint between those two obstacles is superseded by the constraints involving the third. This is analogous to our idea of minimum pathways in the rectilinear case.

Thus, we perform the same reduction: the relevant constraints form precisely a Gabriel graph of the obstacle centers, which is a subgraph of the Delaunay triangulation. Since the Delaunay triangulation can be computed in $\mathcal{O}(n \log n)$ time and provides a valid planar partitioning of the space, we obtain an efficient algorithm for preprocessing the obstacles. The resulting partition has the property that any path between two points is feasible for a robot if and only if its radius is smaller than the minimum constraint along that path.

Thus, we give an analogous result to our main algorithm; the formal proofs of correctness and full details are given in Appendix \ref{app:circ}. This is essentially the intuition we will use for the full problem, which we will now discuss.

\subsection{Constraints induced by rectilinear obstacles of varying shapes}
We now return to the original problem, which concerns rectilinear obstacles and a rectangular robot. We first note the following lemma, which allows us to transform our set of rectilinear obstacles into a set of disjoint rectangular obstacles.

\begin{lemma}
    A rectilinear object consisting of $n$ points can be decomposed into at most $\mathcal{O}(n)$ disjoint rectangles in $\mathcal{O}(n \log n)$ time.
\end{lemma}
\begin{proof}
    For a rectilinear object, we can perform a sweep line from top to bottom, dividing the object into strips of rows first. We would only cut along $y = c$ for some constant $c$ if there exists a point with $y$-coordinate $c$. Since there are at most $n$ points, we make at most $n$ cuts. 
    
    Likewise, within each strip of row, by construction, they should form a set of rectangles with the same height as the row, although they might be discontinuous. We append each of them into our set of decomposed rectangles. We note that at each discontinuity, there must be at least one point at the corner. 
    
    In addition, each point of the rectilinear object is part of at most $4$ discontinuities. Hence, the number of total discontinuities is bounded by $\mathcal{O}(n)$. Since each row and discontinuity generates an extra rectangle, our set consists of at most $\mathcal{O}(n)$ rectangles. 
    
    Finally, we note that the only algorithmic complexity comes from sorting the coordinates and bucketing the points into their respective rows. All of these operations can be done in $\mathcal{O}(n \log n)$ time.
\end{proof}

From here on, we assume we have finished pre-processing all the obstacles and that all of our obstacles are rectangular. We now dive into the main theorem of our paper.

\begin{theorem}
    Given a set of $n$ obstacles $\mathcal{O} = \{O_1, ..., O_n\}$ in the two-dimensional Euclidean space $\mathbb{R}^2$, where each obstacle $O_i$ has center $(x_i, y_i) \in \mathbb{R}^2$ and width $w_i$ and height $h_i$, we can construct a valid partitioning and its planar graph dual.
\end{theorem}

First we prove a useful lemma.

\begin{lemma}
    \label{lem:rect_irrelevant}
    Irrelevant constraints never affect the feasibility of our queries.
\end{lemma}

\begin{proof}
    We first note that a robot's ability to pass through between two rectangles can be interpreted as a generalized $L_{\infty}$ norm since a gap's constraint is the maximum distance along either dimension. A robot of size $s$ can pass through obstacles $O_i$ and $O_j$ if $s \leq \norml{O_i - O_j}$, where we slightly readjust the definition of $\norml{O_i - O_j}$ to be
    \[\norml{O_i - O_j} := \max(|x_i - x_j| - (w_i + w_j)/2, |y_i - y_j| - (h_i + h_j)/2).\]
    
    Next, we show that if an obstacle $O_k$ exists within the minimum pathway of $O_i$ and $O_j$, then we can safely ignore that irrelevant constraint. We first assume that the edges formed by $O_i$ and $O_k$ and $O_j$ and $O_k$ exist. Then, the edge formed by $O_i$ and $O_j$ would not affect our query if and only if the two other constraints are more restrictive than this constraint and if the region formed by $O_i$, $O_j$, and $O_k$ can be safely ignored in our sequence of regions (either because it's always feasible or never feasible).

    We observe that $O_k$ is always closer to both $O_i$ and $O_j$ than they are to each other. This is because if the constraint is of width $s$, then the minimum pathways must span at most $s$ in the other dimension by construction. More specifically, if the thin edge has width $s$ and some height $h \geq 0$, then $O_k$ is at most $s - h$ away from the other two squares in the other direction. Let $s_{i,j}$ denote the generalized $L_{\infty}$ norm of $\norml{O_i - O_j}$; then, we have that $s_{i,j} \geq s_{i,k},s_{j,k}$. This means that in terms of bounding the robot below the $O_i \rightarrow O_j \rightarrow O_k$ boundary, the $s_{i,j}$ constraint would never matter. If $s_{i,j}$ becomes bounding, then $s_{i,k}$ and $s_{j,k}$ would have already been effective.

    This leaves us with the region formed by $O_i, O_j, O_k$. We show that this region is never relevant.
    
    Let the partition region defined by these points be $R$, then for any given queries, there are two cases:
    \begin{itemize}
        \item The starting or ending point falls within $R$: Since we assume that it's always feasible to place the robot at the starting and ending point, the largest the robot can be is a square of side length $s_{i, j}$. Otherwise, the robot would violate the $s_{i,j}$ constraint. However, in this case, the $s_{i,j}$ would never have been bounding since the robot can always pass through this edge.
        \item The path simply passes by region $R$: We only care about the $s_{i, j}$ constraint if the robot is larger than $s_{i, j}$. As we have shown earlier,  $s_{i,j} \geq s_{i,k}, s_{j,k}$. This implies that if the robot is larger than $s_{i, j}$, then it can never enter the region to begin with, and this would be covered by the other two $s_{i, k}$ and $s_{j,k}$ checks. Otherwise, the robot would have to just enter the region through $S_{i, j}$ and leave through $S_{i, j}$ again, which is never optimal. Hence, we can safely ignore the new region created by the irrelevant $s_{i,j}$ edge.
    \end{itemize}

    The two cases show that we never need to treat $R$ as a separate region. Hence, we can safely delete the edge between $O_i$ and $O_j$ and still preserve a valid partitioning.

\end{proof}

However, since we modified the definitions, the resulting graph can no longer be obtained through standard constructions like Gabriel graphs. To show that the resulting graph with these edges is still planar, we claim that two edges can never cross each other.

\begin{lemma}
    \label{lem:planar_constraints}
    For obstacles $O_1, O_2, O_3, O_4$ such that $O_1$ and $O_2$ have a relevant constraint edge and $O_3$ and $O_4$ also have a relevant edge, the Euclidean edge connecting $O_1$ and $O_2$ can never cross $O_3$ and $O_4$.
\end{lemma}

\begin{proof}
As we have mentioned, there are three types of relevant edges. We observe that it is never optimal for the obstacles to have volumes, as it can only increase the chance that an obstacle falls within another obstacle's minimum pathway, thus annulling one of the edges. Hence, we can safely assume that for the diagonal case, the two obstacles are simply points. Furthermore, we combine the partial and complete overlap cases by reducing the rectangles to just two lines covering their intersection. This is valid since it can only improve other obstacles' ability to form edges.

Hence, we are left with three cases, illustrated in Figure \ref{fig:edge_intersections}.
\begin{figure}[h!]
    \centering
    \begin{tikzpicture}[scale=1.25]
        \begin{scope}[shift={(3,0)}]
            \node at (0,2.25) {Case 2a};
            \draw[dashed, draw=orange!60!black, fill=orange!10, line width=0.75pt] (1,0) rectangle ++(-2,2) node[midway] {$S$};
            \draw[edge] (1,0) -- (0,2);
            \draw[edge] (0.05, 0.6) -- (0.9, 1);
            \draw[obs, ultra thick] (1,0) circle (1pt);
            \draw[obs, ultra thick] (0,2) circle (1pt);
            \draw[obs, ultra thick] (0.05,0.6) circle (1pt);
            \draw[obs, ultra thick] (0.9,1) circle (1pt);
            
        \end{scope}
        \begin{scope}[shift={(6,0)}]
            \node at (0,2.25) {Case 2b};
            \draw[dashed, draw=orange!60!black, fill=orange!10, line width=0.75pt] (1,0) rectangle ++(-2,2) node[midway] {$S$};
            \draw[edge] (1,0) -- (0,2);
            \draw[edge] (0.05, -0.2) -- (1.2, 1);
            \draw[obs, ultra thick] (1,0) circle (1pt);
            \draw[obs, ultra thick] (0,2) circle (1pt);
            \draw[obs, ultra thick] (0.05,-0.2) circle (1pt);
            \draw[obs, ultra thick] (1.2,1) circle (1pt);
            
        \end{scope}
        \begin{scope}[shift={(0,0)}]
            \node at (0.5,2.25) {Case 1};
            \draw[edge] (0.3,0) rectangle (0.6,2);
            \draw[edge] (0.1,1) rectangle (0.9, 1.5);
            \draw[obs,ultra thick] (0.1,1) -- (0.1,1.5);
            \draw[obs,ultra thick] (0.9,1) -- (0.9,1.5);
            \draw[obs,ultra thick] (0.3,2) -- (0.6,2);
            \draw[obs,ultra thick] (0.3,0) -- (0.6,0);
        \end{scope}
        \begin{scope}[shift={(8,0)}]
            \node at (0.5,2.25) {Case 3a};
            \draw[edge] (0, 1) rectangle (1, 1.3);
            \draw[edge] (0.25,1.8) -- (0.75,0.5);
            \draw[obs, ultra thick] (0.25,1.8) circle (1pt);
            \draw[obs, ultra thick] (0.75,0.5) circle (1pt);
            \draw[obs,ultra thick] (0,1) -- (0,1.3);
            \draw[obs,ultra thick] (1,1) -- (1,1.3);
        \end{scope}
        \begin{scope}[shift={(10,0)}]
            \node at (0.5,2.25) {Case 3b};
            \draw[edge] (0, 1) rectangle (1, 1.3);
            \draw[edge] (-0.5,1.2) -- (1.5,1.1);
            \draw[obs, ultra thick] (-0.5,1.2) circle (1pt);
            \draw[obs, ultra thick] (1.5,1.1) circle (1pt);
            \draw[obs,ultra thick] (0,1) -- (0,1.3);
            \draw[obs,ultra thick] (1,1) -- (1,1.3);
        \end{scope}
    \end{tikzpicture}

    \caption{The cases of edge intersections.}
    \label{fig:edge_intersections}
\end{figure}
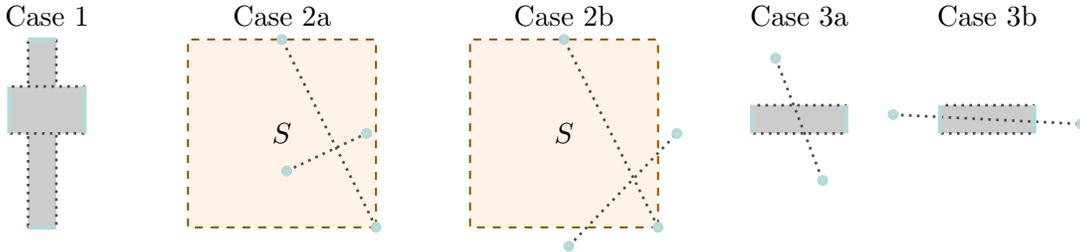

\begin{itemize}
    \item Both the $O_1$ and $O_2$ edge and $O_3$ and $O_4$ edge are axis-aligned. We can decompose each edge into their two bounding lines. We observe that for the two thin edges to intersect, at least one out of the four pairs of bounding lines must intersect. Hence, we reduce our problem to assuming that $O_1, O_2, O_3, O_4$ are all point obstacles too. We note that this then becomes a subset of the case below, where both edges are point obstacles that are ``diagonal'' to each other.
    \item Both the $O_1$ and $O_2$ edge and $O_3$ and $O_4$ edge are of diagonal type. We again denote $s_{1,2} = \norml{O_1 - O_2}$ and $s_{3,4} = \norml{O_3 - O_4}$. Without the loss of generality, let $s_1 \geq s_2$. Let us imagine that $O_1$ is on the top left and $O_2$ is on the bottom right. We first focus on $O_2$. 
    
    Let $S$ be a square with width $s_{1,2}$ and with $O_2$ as the bottom right corner, then this square is within the minimum pathway. We now have two cases. If either $O_3$ and $O_4$ are within the square, then we are done. Otherwise, there are two cases in order for an intersection to occur. 
    
    If $O_3$ and $O_4$ are on opposing sides of the square, this is a contradiction since then $\norml{O_3 - O_4} > \norml{O_1 - O_2}$ (note that since the minimum pathway is inclusive on edge points, equality cannot happen). Otherwise, $O_3$ and $O_4$ are on adjacent sides, more specifically, either on the sides of $O_1$ or $O_2$. We assume they are next to $O_2$, as repeating this proof analogously with $O_1$ and $O_2$ clipped should cover all cases. Then, for the edges to intersect, one point has to be higher than $O_2$ and another point has to be to the left of $O_2$. This necessarily means that the rectangle formed by endpoints $O_3$ and $O_4$ contains the $O_2$ corner of the square, a contradiction.
    \item Without the loss of generality, the $O_1$ and $O_2$ edge is diagonal, but the $O_3$ and $O_4$ edge is axis-aligned. We again decompose the $O_3$ and $O_4$ edge into the two bounding lines. If the $O_1$ and $O_2$ edges intersect either bounding lines, then our diagonal case above finishes the proof. The only remaining case is when the diagonal edge is between them. If either points are within the axis-aligned pair's thin edge, then we are done. Otherwise, they must both be outside on opposing sides. However, their minimum pathway would then necessarily intersect with $O_3$ and $O_4$. 
\end{itemize}
Hence we have exhausted all cases, so we are done.
\end{proof}

Since no two edges ever intersect, this implies that our resulting graph is planar. We have also shown that by only taking relevant edges, our graph is a valid region partitioning scheme. Finally, we introduce an algorithm to efficiently identify these relevant edges.

\subsection{Finding a region partitioning scheme for rectilinear obstacles}
\label{sec:partition_algorithm}
\begin{theorem}
    We can find all relevant edges and construct the region partitioning and dual planar graph in $\mathcal{O}(n \log n)$ time using a sweep line algorithm.
\end{theorem}

To make an algorithm, we will need to create some sort of precedence for our obstacles; we do this with the notation of \textit{shadow regions} as described below.

\begin{define}
    \label{def:shadow_region}
    For a given obstacle $O$ described by $(x_1, x_2) \times (y_1, y_2)$, we define the \textbf{shadow region} as the trapezoid bounded by the lines $y = y_1$, $x = x_1$, and $y = y_2 + (x_1 - x)$. (An illustration of this can be found in Figure \ref{fig:shadow_region}.)
\end{define}

\begin{figure}[h!]
    \centering
    \begin{tikzpicture}[scale=0.5]
        \fill[pattern=crosshatch dots, pattern color=gray] (-9,0) -- (5,0) -- (5,3) -- (-2,10) -- (-9, 10) -- cycle;
        \draw[obs] (5,0) rectangle (12,3) node[midway] {\large $O_j$};
        \draw[obs] (-8,4) rectangle (2,9) node[midway] {\large $O_i$};
        \draw[red, line width=1.5pt] (-8,4) -- (2,4);
        \draw[dashed, line width=0.75pt] (5,3) -- (-2,10);
        \draw[dashed, line width=0.75pt] (5,0) -- (-9,0);
        \node[fill=white, rounded corners] at (-2,2) {shadow region};
    \end{tikzpicture}
    \caption{Shadow region of $O_j$ (where $O_i$ is part of it).}
    \label{fig:shadow_region}
\end{figure}
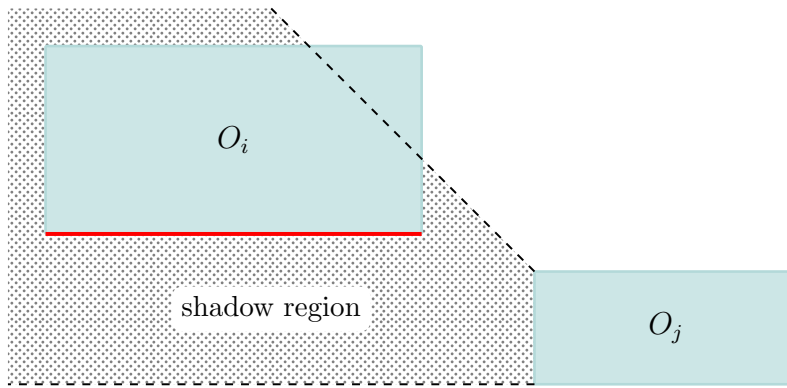

We process our obstacles from left to right. In addition, for obstacle $O_i$, we focus on its connections to other obstacles within the shadow region for now. We claim that once an obstacle falls within another obstacle's shadow region, it can be safely discarded for future obstacles.

\begin{lemma}
    \label{lem:shadow_region}
    If we process the set of obstacles in ascending order of $x_1$, then for a given obstacle $O_i$, all obstacles $O_j$ such that their bottom side $y = y_1$ and $x \in [x_1, x_2]$ is completely within the shadow region become irrelevant for future obstacles. More specifically, for obstacle $O_k$ processed after $O_i$, $O_k$ should never need to form edges with obstacles are also in $O_i$'s shadow region.
\end{lemma}

\begin{proof}
    We first note that $O_j$ is above $O_i$ (where above/below is measured by the bottom side). We then note that if obstacle $O_k$ is below $O_i$, then since $O_j$ is above $O_i$, we get $O_j \geq O_i \geq O_k$. Consequently, the minimum pathway between $O_k$ and $O_j$ will necessarily intersect $O_i$. If $O_k$ is above $O_j$, then $O_j$ will not be in $O_k$'s shadow region, so we do not care about this case. Hence, the only case left is when $O_j \geq O_k \geq O_i$, illustrated in Figure \ref{fig:shadow_region_ordering}. 
    
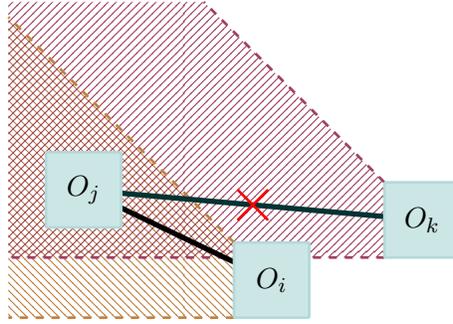
\begin{figure}[h!]
    \centering
    \begin{tikzpicture}[decoration={
    markings,mark=at position .5 with \node[red]{\Huge $\times$};}]
    \fill[pattern=north west lines, pattern color= orange!50!gray] (0,0) -- (0,1) -- (-3,4) -- (-3,0) -- cycle;

    \fill[pattern=north east lines, pattern color=purple!50!gray] (2,0.8) -- (2,1.8) -- (-0.4,4.2) -- (-3,4.2) -- (-3,0.8) -- cycle;
    
        \draw[dashed, orange!50!gray, line width=1pt] (0,0) -- (-3, 0);
        \draw[dashed, purple!50!gray, line width=1pt] (2,0.8) -- (-3,0.8);
        \draw[dashed,  orange!50!gray, line width=1pt] (0,1) -- (-3,4);
        \draw[dashed, purple!50!gray, line width=1pt] (2,1.8) -- (-0.4,4.2);
        
        \draw[black, line width=2pt] (0.5,0.5) -- (-2,1.7);
        \draw[teal!40!black, line width=2pt, postaction={decorate}] (-2,1.7) -- (2.5,1.3);
        
        \draw[obs] (0,0) rectangle ++(1,1) node[midway] {$O_i$};
        \draw[obs] (-2.5,1.2) rectangle ++(1,1) node[midway] {$O_j$};
        \draw[obs] (2,0.8) rectangle ++(1,1) node[midway] {$O_k$};
    \end{tikzpicture}

    \caption{The case when $O_j \geq O_k \geq O_i$.}
    \label{fig:shadow_region_ordering}
\end{figure}

    However, we claim that the minimum pathway between $O_j$ and $O_k$ will always contain $O_i$, thus rendering the edge irrelevant. Due to the upper slope of the shadow region (roughly $y = x$), the bottleneck width will always be along the $x$ direction. As in
    $$ \norml{O_j - O_k} = x^k_1 - x^j_2 $$
    Likewise, the same applies to $O_j$ and $O_i$
    $$ \norml{O_j - O_i} = x^i_1 - x^j_2 \geq y^j_1 - y^i_2 $$
    However, the minimum pathway between $O_j$ and $O_k$ contains a square with $O_j$ as the top left corner and side length $x^k_1 - x^j_2$. Since we process the squares by $x_1$, this quantity must be greater than or equal to $x^i_1 - x^j_2$. Hence, the triangle extends below to the following $y$-coordinate
    $$ y^j_1 - (x^k_1 - x^j_2) \leq y^j_1 - (x^i_1 - x^j_2) \leq y^j_1 - (y^j_1 - y^i_2) = y^i_2 $$
    Hence, the minimum pathway must extend below the top of $O_i$, rendering the edge irrelevant.
\end{proof}

\begin{cor}
    There can be at most $n$ edges from the outlined procedure.
\end{cor}

\begin{proof}
Our algorithm processes the obstacles in ascending $x_1$ order. Then, for a given obstacle $O_i$, we find all obstacles $O_j$ within the shadow region and add an edge between $O_j$ and $O_i$. Finally, we delete all of those obstacles. Since every time we add an edge, we delete an obstacle, there are at most $n$ total edges.
\end{proof}

Implementation-wise, we can execute this algorithm using a sweep line algorithm. We claim that obstacles within the shadow region must have consecutive $y_1$ values. Assume otherwise, when processing $O_i$, it includes $O_j$ but not $O_k$ where $y^j_1 > y^k_1$. If $x^j_1 \leq x^k_1$, then processing $O_k$ should have deleted $O_j$ already as it would be in the shadow region of $O_k$. Otherwise, the scenario is impossible since the shadow region is monotonically expanding with respect to $x$. Thus, we can keep track of an ordered set (using your balanced binary search tree of choice) storing pairs of $(y^j_1, O_j)$. When processing $x_i$, we find the lower bound with $y^i_1$ and keep deleting obstacles and adding edges until it exceeds our threshold. Each operation takes $\mathcal{O}(\log n)$ time, so the total runtime should be $\mathcal{O}(n \log n)$.

However, for a given obstacle, our algorithm currently only considers edges within the shadow region, which is only about one-eighth of the entire plane of possibilities. However, there's an easy fix — we repeat the algorithm eight times by rotating the plane in all four directions and flipping upside down too. One can verify that this should cover all possible edges, and there are at most $8n$ edges in the end.

Nevertheless, there is one more condition we must satisfy: the planar property. We have shown that the edges we found must contain all relevant edges, but some of them can be irrelevant, too, which would mess up our planar graph criterion. To resolve this issue, we do an extra cleanup step at the end. Namely, we check for each edge if it is relevant. To do this efficiently, we first find the minimum pathway rectangle for each edge. Then, determining their relevance is equivalent to querying if the rectangle intersects with any obstacles. Since our obstacles are fixed, we can answer these queries off-line. This is a standard computational geometry problem that can be done in $\mathcal{O}(n \log n)$ time. A sketch of such an algorithm is as follows:
\begin{enumerate}
    \item Perform coordinate compression such that all coordinates are integers between $[0, 2n - 1]$.
    \item Sweep lining from left to right. If we encounter a left side, we range increment $[y_1, y_2]$, and likewise decrement for the right side.
    \item To answer queries, we check if the sum is $1$ at the left side. Otherwise, we append this interval to another ordered set. Whenever we process other rectangles, we first check if it intersects with any queries. If it does, we mark those edges as irrelevant and delete them from our set.
\end{enumerate}

We omit the proof of correctness here as it is quite standard, but one can verify that it works in $\mathcal{O}(n \log n)$ time. Furthermore, maintaining these data structures should each take $\mathcal{O}(n)$ space. 

This gives us a list of all relevant edges. As we have proven in Lemma \ref{lem:planar_constraints}, the resulting edges form a planar graph, and, therefore, a valid partitioning scheme. This concludes our algorithm, which runs in $\mathcal{O}(n \log n)$ time and $\mathcal{O}(n)$ space. 

Finally, we introduce a data structure to answer our online queries in $\mathcal{O}(\log n)$ time.

\subsection{Answering queries in $\mathcal{O}(\log n)$ with persistent disjoint-set union}
\label{sec:persistent_dsu}
We introduce a data structure on top of our dual planar graph for answering queries in $\mathcal{O}(\log n)$ time and $\mathcal{O}(n)$ space. Recall from our formulation in Theorem \ref{thm:constraint} that to answer a query in the form of $(s, t, r)$, where a robot of size $r$ starts at $s$ and wants to reach destination $t$, it is equivalent to checking there exists a sequence of regions such that the first region contains $r$, the last region contains $t$, and our robot can fit through all the edges that connect these regions. This is equivalent to checking whether a path exists in the dual planar graph such that the maximum capacities of edges along the path are all at least as large as $r$. In other words, on the induced sub-graph that only contains edges of capacity $r$ or greater, we need to determine if the nodes corresponding to the regions of $s$ and $t$ are in the same connected component. We claim this can be maintained using a persistent disjoint-set union set data structure.

\begin{theorem}
    Given a region partitioning scheme and its dual planar graph, we can answer online queries of whether a robot of size $r$ can start at point $s$ and reach destination point $t$ in $\mathcal{O}(\log n)$ and $\mathcal{O}(n)$ total space.
\end{theorem}

\begin{proof}
    We maintain a persistent disjoint-union set data structure that supports the following operations:
    \begin{itemize}
        \item \textsc{Union} nodes $u$ and $v$ into the same component at time step $t$, which increments by one for each union operation.
        \item \textsc{Query} whether nodes $u$ and $v$ are part of the same component at a given time step $t$.
    \end{itemize}
    As given in \cite{DavidKarger}, each query to our persistent DSU takes $\mathcal{O}(\log n)$ time and $\mathcal{O}(n)$ total space, where $n$ is the number of nodes. We first sort the list of edges in our dual planar graph in descending order of their capacity. For some edge $(u, v)$, we then add back the edge by performing a union operation on nodes $u$ and $v$. This should take $\mathcal{O}(n \log n)$ preprocessing time since (as we have shown earlier) the planar graph has $\mathcal{O}(n)$ edges. We also keep a sorted list of all of the edges' capacities and their corresponding union time steps.

    Now, we can answer queries of the form $(r, s, t)$ in $\mathcal{O}(\log n)$ time. We first find the nodes $u$ and $v$ corresponding to $s$ and $t$. Then, we find the largest time step in our sorted list such that the corresponding maximum capacity is at least $r$, which can be done in $\mathcal{O}(\log n)$ time with binary search. Finally, we check if $u$ and $v$ are part of the same component. If they are, then that means that there exists a path that only traverses edges with capacities of at least $r$, as desired. This concludes our algorithm. 
\end{proof}

Finally, we combine this with our sweep line algorithm from Section \ref{sec:partition_algorithm}, which allows us to pre-process everything in $\mathcal{O}(n \log n)$ time and $\mathcal{O}(n)$ space. Then, we can answer online feasibility queries for varying robot sizes in $\mathcal{O}(\log n)$ time, as desired.

\section{Experiments}
\begin{figure}[h!]
    \centering
    \includegraphics[width=0.28\linewidth]{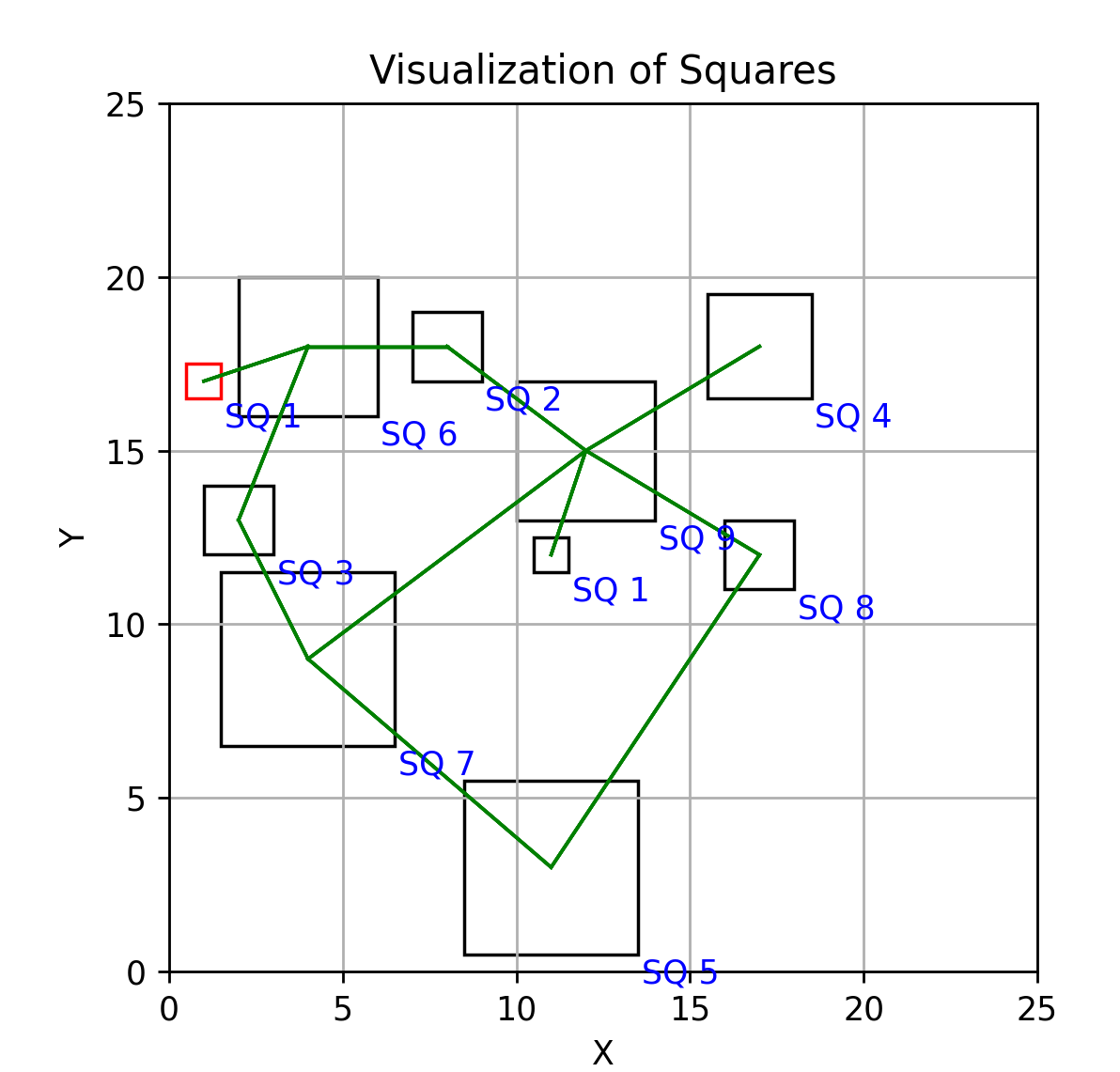}
    \includegraphics[width=0.3\linewidth]{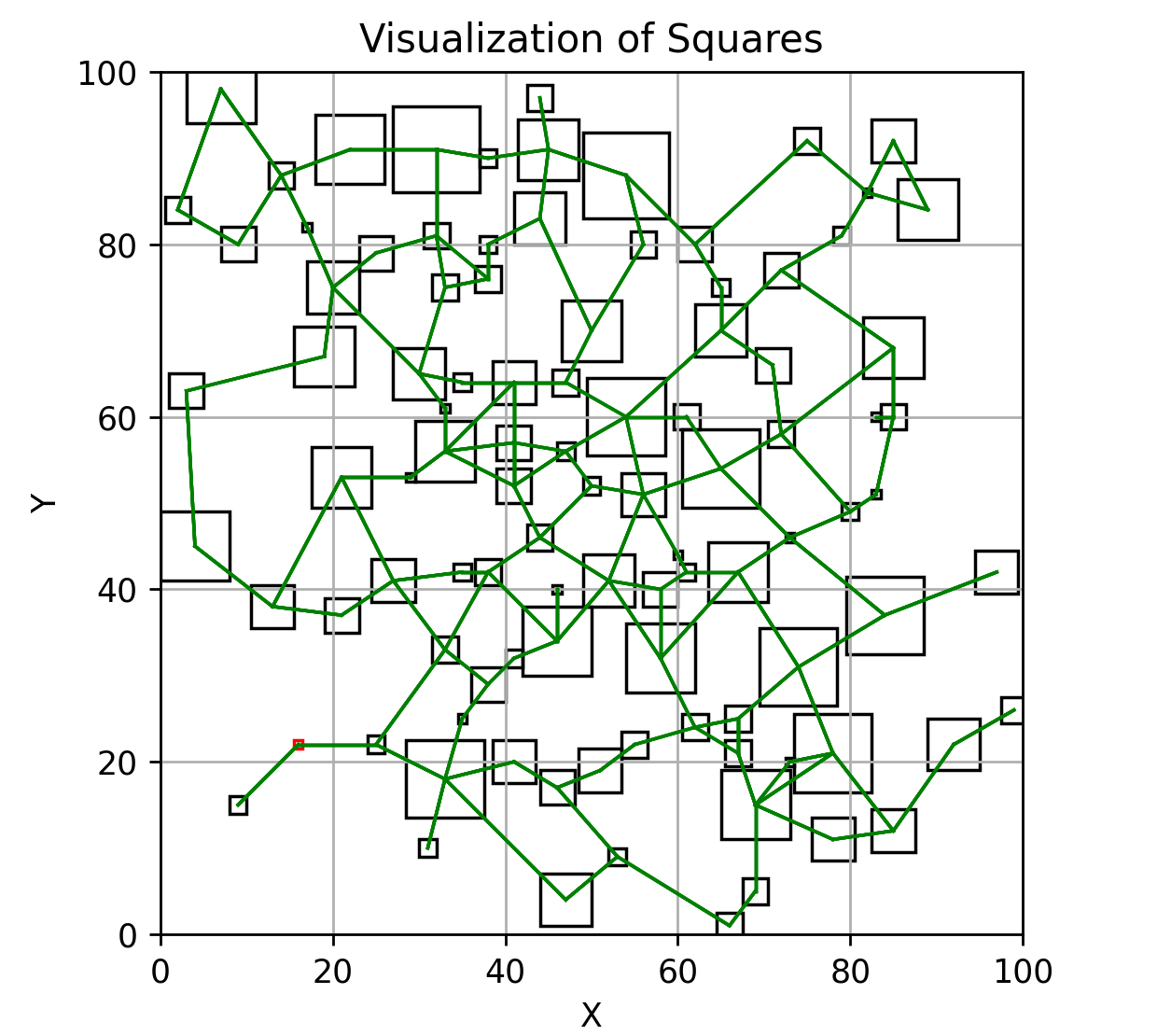}
    \includegraphics[width=0.3\linewidth]{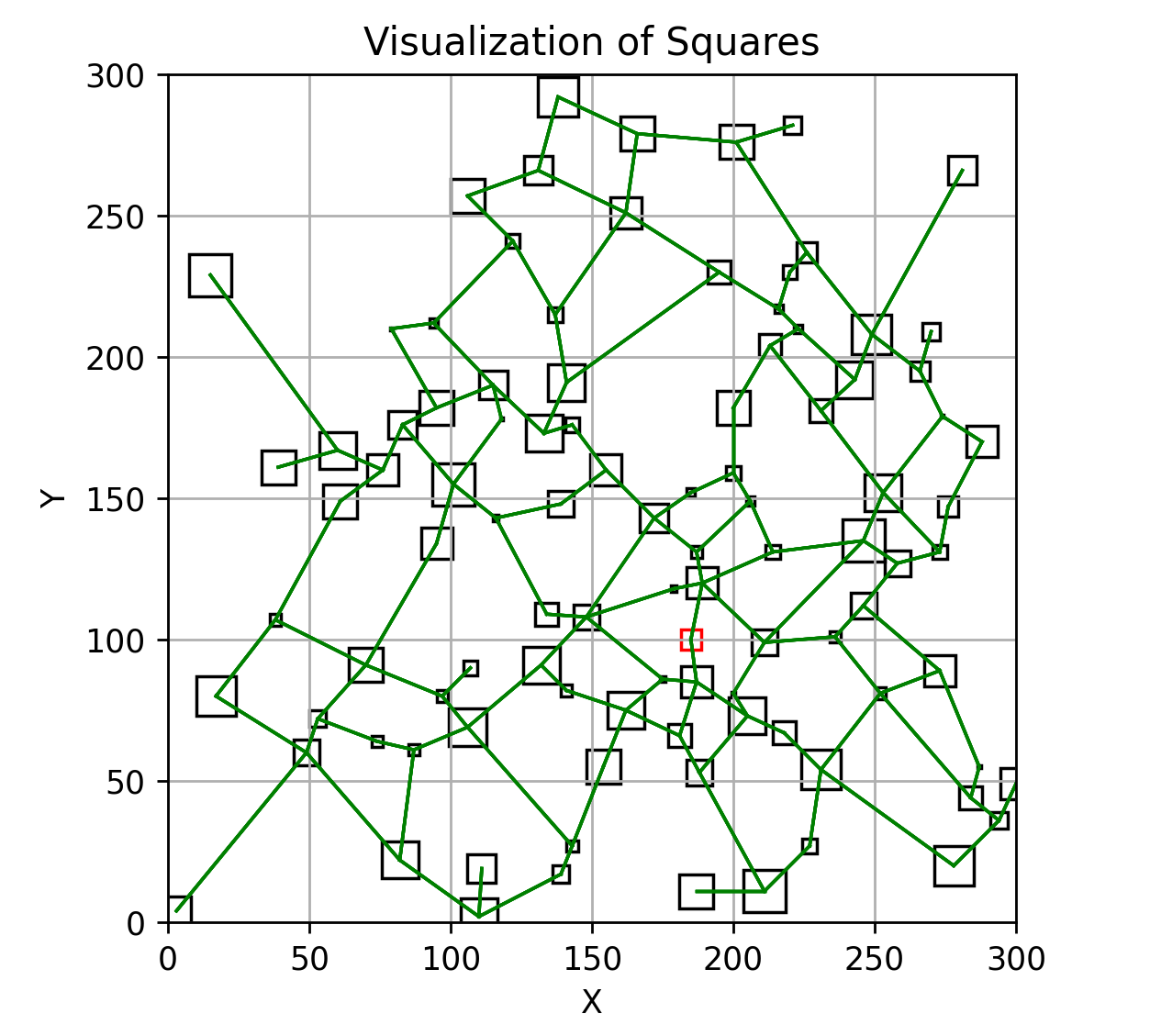}
    \caption{Examples of our partitions on various rectilinear obstacle maps.}
    \label{fig:enter-label}
\end{figure}
We implemented our algorithm in C++ and tested it against random sets of obstacles. In addition, we implemented a slower algorithm that explores all possible paths using a DFS-based pathfinding algorithm to verify correctness. Due to time constraints and the ease of generation, our testcases only contain squares. After countless runs of various parameters and landscape generation schemes (completely random, clustering, and mazes, etc.), we found no discrepancy between the results generated by the two algorithms.


\section{Conclusions and Future Directions}
Our paper's main contribution lies in our introduction of the generalized Gabriel graph framework. We offer both an intuitive and formalized definition for partitioning graphs into relevant edge boundaries and regions, as well as their dual planar graph representation. Furthermore, we demonstrate the framework's usefulness by describing an algorithm that efficiently partitions a space of rectilinear obstacles using the generalized Gabriel graph technique in $\mathcal{O}(n \log n)$ time and $\mathcal{O}(n)$ space. We also solve the problem of circular obstacles with the same radii and demonstrate how we can obtain the Gabriel graph using Delaunay Triangulation in $\mathcal{O}(n \log n)$ time. Finally, we present a persistent-DSU-based data structure, which enables us to answer online path feasibility queries in $\mathcal{O}(n \log n)$ time.

Beyond our work here, our framework has much more to offer. For example, we believe that a generalized Gabriel graph definition should exist for circular obstacles with different radii, but with the caveat that the boundaries might be non-linear. Likewise, the same idea can be applied to other shapes and formulations, as long as one can find the appropriate definitions for edges, regions, and minimum pathways.

In addition, by breaking down the space into regions and constructing a dual planar graph, our framework empowers us to perform other common queries as well. For example, we can run the shortest path problem on our dual graph. While this does not account for intra-regional movement, regions are highly flexible and easy to preprocess. Hence, we believe with careful maneuvering, our framework can be leveraged to solve many other complex challenges.

\section{Acknowledgments}
We would like to thank Professor David Karger for teaching the Advanced Algorithms Course at MIT.

\bibliographystyle{abbrv}
\bibliography{bib}

\begin{appendix}
\section{The Minkowski Sum Technique} \label{app:mink}
An equivalent formulation of our motion planning problem can be obtained through what we refer to as the \textit{Minkowski sum technique}, visually shown in Figure \ref{fig:mink}. This transformation allows us to simplify the problem by treating the robot as a point while appropriately expanding the obstacles. Below, we formally define this transformation and prove its equivalence to the original problem formulation.

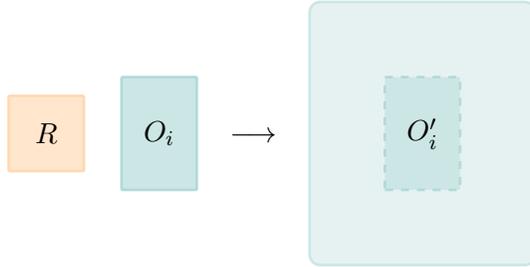
\begin{figure}[h!]
    \centering
    \begin{tikzpicture}
        \draw[robot] (0,0.75) rectangle (1,1.75) node[midway] {$R$};
        \draw[obs] (1.5,0.5) rectangle (2.5,2) node[midway] {$O_i$};
        \node at (3.25,1.2) {$\longrightarrow$};
        \draw[fill=teal!10,draw=teal!20,rounded corners,line width=1 pt] (4,-0.5) rectangle (7,3);
        \draw[obs,dashed] (5,0.5) rectangle (6,2) node[midway] {$O'_i$};
    \end{tikzpicture}
    \caption{A visualization of the Minkowski sum technique.}
    \label{fig:mink}
\end{figure}

Given an obstacle $O_i$, we define its \textit{expanded version} $O_i'$ as the Minkowski sum of $O_i$ and the robot $R$ centered at the origin. Given that the robot is a square and the obstacles are rectangles, the expanded obstacle $O_i'$ is also an axis-aligned rectangle with corners $(x_1^i - \frac{d}{2}, y_1^i - \frac{d}{2})$ and $(x_2^i + \frac{d}{2}, y_2^i + \frac{d}{2})$. Specifically, $O_i'$ occupies the region $[x_1^i - \frac{d}{2}, x_2^i + \frac{d}{2}] \times [y_1^i - \frac{d}{2}, y_2^i + \frac{d}{2}]$. We claim that the robot can now be treated as a point.

\begin{theorem}
A configuration $(x,y)$ is collision-free for the square robot $R$ with respect to obstacle $O_i$ if and only if the point $(x,y)$ does not lie within the expanded obstacle $O_i'$.
\end{theorem}

\begin{proof}
($\Rightarrow$) Suppose configuration $(x,y)$ is collision-free with respect to $O_i$. This means that $R(x,y) \cap O_i = \varnothing$. For the sake of contradiction, assume $(x,y) \in O_i'$. By the definition of Minkowski sum, this means there exist points $p \in O_i$ and $q \in R(0,0)$ such that $(x,y) = p + q$. But this implies that $p = (x,y) - q$ is a point that lies in both $O_i$ and $R(x,y)$, contradicting our assumption that the configuration is collision-free.

($\Leftarrow$) Suppose $(x,y) \not\in O_i'$. Again, for the sake of contradiction, assume $R(x,y) \cap O_i \neq \varnothing$. Then there exists some point $p$ that lies in both $R(x,y)$ and $O_i$. Let $q = (x,y) - p$. By construction, $q \in R(0,0)$ since $R(x,y)$ is just a translation of $R(0,0)$ by vector $(x,y)$. Therefore, $(x,y) = p + q$ where $p \in O_i$ and $q \in R(0,0)$, implying $(x,y) \in O_i'$, which contradicts our assumption.
\end{proof}

Thus, we can transform our original motion planning problem into an equivalent problem where we need to find a path for a point robot among enlarged obstacles. One advantage of this transformation is that it simplifies collision checking; instead of computing intersections between squares, one can simply consider if a point lies within an obstacle in $\mathcal{O}$.

\section{A Proof of Constraint Equivalence} \label{app:constr}
Before proving constraint equivalence, we establish several important lemmas about our region partitioning.

\begin{lemma}
\label{lem:region-connectivity}
For any region $R$ in our partitioning, given any two points $p, q \in R$, there exists a continuous collision-free path connecting them within $R$.
\end{lemma}

\begin{proof}
By the construction of our regions, each region $R$ is a connected component of the plane bounded by obstacles and thin edges. Since obstacles and thin edges are removed from the region, $R$ is an open set. Therefore, for any two points $p, q \in R$, we can construct a collision-free path between them by following any continuous curve that stays within $R$. Such a curve always exists because $R$ is path-connected (being an open connected subset of $\mathbb{R}^2$).
\end{proof}

\begin{lemma}
\label{lem:boundary-crossing}
If there exists a collision-free path between any point in region $R_i$ to any point in adjacent region $R_{i+1}$, then there exists a specific point $b$ on their shared boundary such that any point in either region can reach $b$ through a collision-free path.
\end{lemma}

\begin{proof}
The shared boundary between adjacent regions that the robot can pass through is composed of thin edges. By the definition of pairwise passability, there must exist at least one point $b$ on this boundary that is collision-free for the robot. Due to the openness of the regions and the continuity of the configuration space, any point in either region can reach $b$ through a straight-line path (which may need to be slightly perturbed to avoid obstacle corners, but such a perturbation is always possible due to the openness of the regions).
\end{proof}

Now we proceed to the proof of the main theorem.

\begin{proof}
($\Rightarrow$) Suppose there exists a feasible trajectory $\tau$ from $s$ to $g$. Since $\tau$ is continuous and collision-free, it must pass through a sequence of regions $R_1, R_2, ..., R_k$ where $s \in R_1$ and $g \in R_k$. For any consecutive pair of regions $(R_i, R_{i+1})$, the trajectory $\tau$ demonstrates the existence of at least one collision-free path between points in these regions (namely, the subpath of $\tau$ that connects them). By the continuity of the configuration space and the openness of regions, this path can be continuously deformed to connect any other pair of points in these regions while maintaining collision-freedom, thus satisfying the pairwise passability constraint.

($\Leftarrow$) Suppose we have a sequence of regions $R_1, R_2, ..., R_k$ satisfying the conditions. We can construct a feasible trajectory from $s$ to $g$ as follows:

\begin{enumerate}
\item Start at $s \in R_1$.
\item For each consecutive pair of regions $(R_i, R_{i+1})$:
\begin{itemize}
   \item By Lemma \ref{lem:boundary-crossing}, find a boundary point $b_i$ between $R_i$ and $R_{i+1}$.
   \item By Lemma \ref{lem:region-connectivity}, construct a path from the current position to $b_i$ within $R_i$.
   \item Then, move to $R_{i+1}$ through $b_i$.
\end{itemize}
\item Finally, by Lemma \ref{lem:region-connectivity}, construct a path from the last boundary point to $g$ within $R_k$.
\end{enumerate}

The concatenation of these paths forms a feasible trajectory from $s$ to $g$. The trajectory is continuous because each segment is continuous and they connect at the boundary points. It is collision-free because each segment either stays within a single region (where Lemma \ref{lem:region-connectivity} guarantees collision-freedom) or crosses between regions through a proven-safe boundary point.
\end{proof}

\section{The Circular Case} \label{app:circ}

Instead of rectilinear objects, suppose each query defines a robot as a circle of radius $q_i$, starting point $s_i$, and ending point $t_i$. In addition, all obstacles are circles with the same radius $r$. More formally, we prove the following theorem:
\begin{theorem}
    Given a set of $n$ obstacles $\mathcal{O} = \{O_1, ..., O_n\}$ in the two-dimensional Euclidean space $\mathbb{R}^2$, where each obstacle $O_i$ is characterized by its center $(x_i, y_i) \in \mathbb{R}^2$ and shared radius $r$, we can construct a valid partitioning and its planar graph dual in $\mathcal{O}(n \log n)$ runtime. 
\end{theorem}

\begin{proof}
    We invoke the Minkowski sum technique of Appendix \ref{app:mink} in reverse (i.e. shrink the obstacles into points and embiggen the robot, which is valid since all obstacles are congruent). That is, for a query of a robot with radius $q_i$, then we instead replace it with a robot of radius $q_i + 2r$. We can verify that if the original robot can pass through obstacles $O_i$ and $O_j$, then this implies $$q_i \leq \norm{O_i - O_j} + 2r \implies q_i + 2r \leq \norm{O_i - O_j},$$
    as desired. The key intuition behind our algorithm is that each pair of obstacles induces an extra constraint on the feasibility of possible walks. However, most of these edges are irrelevant, as their constraints are strictly less powerful than a combination of other constraints. Namely, we have the following.

    \begin{lemma}
        \label{lemma_gabriel_graph}
        Irrelevant constraints never affect the feasibility of our queries.
    \end{lemma}

    \begin{proof}
        The obstacles $O_i$ and $O_j$ impose a constraint of maximum diameter $||O_i - O_j||$ for the robot. Furthermore, in order for a robot with that diameter to pass between the two obstacles, there must be no extra obstacles within the circle with $O_i$ and $O_j$ as the two diameter endpoints. Otherwise, if there exists another obstacle $O_k$ in the way, then the maximum constraint can never be tight. This implies that $O_k$ is closer to both $O_i$ and $O_j$ than $O_i$ and $O_j$ are to each other, so the constraints formed by $O_i$ and $O_k$ and $O_j$ and $O_k$ supersedes the constraint formed by $O_i$ and $O_j$.

        Hence, we have shown that if we simply ignore the edge $O_i$ to $O_j$, then the two other constraints is sufficient in bounding the circle below the $O_i  \rightarrow O_j \rightarrow O_k$ line. However, we still need to show that the triangular region constraint formed by $O_i$, $O_j$, and $O_k$ is also irrelevant. We denote $s_{i,j} = ||O_i - O_j||$, then by construction we have $s_{i,j} \geq s_{i, k}, s_{j, k}$. We define a robot of radius $s_{i, j} + \epsilon$. Then, the $O_i$ and $O_j$ constraint effectively bounds the center of any robot with $\epsilon > 0$ below the line segment between $O_i$ and $O_j$. However, we claim that just the $O_i$ and $O_k$ and $O_j$ and $O_k$ segment provides a stricter bound. Consider the case that point $O_k$ does not exist, then with a small enough $\epsilon$, the center of our robot should fit almost exactly on the $O_i$ and $O_j$ segment. However, since, $O_k$ is within the circumcircle by construction, $O_k$ will push the robot further away from the triangular region constraint than if the additional obstacle does not exist. Hence, the two segments $O_i, O_k$ and $O_j, O_k$ is sufficient, meaning that we can just delete the $O_i$ and $O_j$ segment.
        \begin{figure}[h!]
            \centering
            \begin{tikzpicture}
                \draw[black!80, dashed] (0,0) circle (2);
                
                \node[obs] (O1) at (-2,0) {$O_i$};
                \node[obs] (O2) at (2,0) {$O_j$};
                \node[obs] (O3) at (0,1.5) {$O_k$};
                
                \draw[edge] (O1) -- (O2) node[midway, below] {irrelevant};
                \draw[edge] (O1) -- (O3);
                \draw[edge] (O2) -- (O3);
                
                \path[region, opacity=0.2] (O1.center) -- (O2.center) -- (O3.center) -- cycle;
                
                \node[robot] (R) at (0,0.65) {$R$};
            \end{tikzpicture}
            \caption{A visualization of the circumcircle.}
            \label{fig:circumcircles}
        \end{figure}
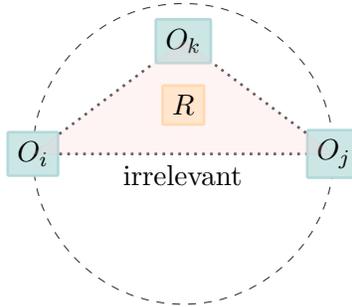

        Finally, we offer an alternative proof regarding the relevancy of triangle $O_i, O_j, O_k$, since the proof above is difficult to generalize for the original problem with rectangular robot and obstacles. 
        
        Let the partition region defined by the triangle be $R$, then for any given queries, there are two cases:
        \begin{itemize}
            \item The starting or ending point falls within $R$: Since we assume that it's always feasible to place the robot at the starting and ending point, the largest the robot can be is the circumcircle of $R$. By construction, $O_k$ is within the circumcircle formed by the diameter with endpoints $O_i$ and $O_j$, therefore it must have a diameter at most as large as $s_{i,j}$. Thus, the constraint induced by $s_i$ and $s_j$ is never effective.
            \item The path simply passes by region $R$: We only care about the $s_{i, j}$ constraint if the robot is larger than $s_{i ,j}$. However, by construction, the line formed by any two points inside the circle is at most as large as the diameter, so $s_{i,j} \geq s_{i, k}, s_{j, k}$.  This implies that if the robot is larger than $s_{i, j}$, then it can never enter the region to begin with, and this would be covered by the other two $s_{i, k}$ and $s_{j,k}$ checks. Otherwise, the robot would have to just enter the region through $S_{i, j}$ and leave through $S_{i, j}$ again, which is never optimal. Hence, we can safely ignore the $s_{i, j}$ constraint. 
        \end{itemize}

        The two cases show that we never need to treat $R$ as a separate region. Hence, we can safely delete the edge between $O_i$ and $O_j$ and still preserve a valid partitioning.
    \end{proof}

    The conditions defined by relevant constraints defined in Definition \ref{def:gabriel_graph} and Lemma \ref{lemma_gabriel_graph} are equivalent to conditions for constructing a Gabriel graph. It is well-known that the Gabriel graph is a subgraph of the Delaunay Triangulation; we repeat the proof below.
    \begin{theorem}
        For a set of $n$ points $O_1, ..., O_n$, the corresponding Gabriel graph, where points $O_i$ and $O_j$ share an edge if and only if no other points lie within the circle with diameter endpoints $O_i$ and $O_k$, is a subgraph of the Delaunay triangulation.
    \end{theorem}
    \begin{proof}
        Let edge $e=(O_i,O_j)$ be an arbitrary edge in the Gabriel graph. By definition of the Gabriel graph, the circle $C$ with diameter $O_iO_j$ contains no other points from the set. 
        
        Namely, the center of $C$ is at the midpoint of $O_iO_j$, and both $O_i$ and $O_j$ lie on its boundary. Now, observe that any circle passing through both $O_i$ and $O_j$ that has center on the perpendicular bisector of $O_iO_j$ and radius greater than or equal to $\frac{|O_iO_j|}{2}$ must completely contain circle $C$. Therefore, there exists a circle passing through $O_i$ and $O_j$ that contains no other points from the set.
        
        By the definition of the Delaunay triangulation, two points share an edge if and only if there exists a circle passing through them that contains no other points from the set. We have shown such a circle exists for edge $e$, therefore $e$ must also be present in the Delaunay triangulation.
        
        Since this holds for any arbitrary edge in the Gabriel graph, we conclude that the Gabriel graph is indeed a subgraph of the Delaunay triangulation.
    \end{proof}
    
    Since the triangulation is itself planar, all of our desired properties of the dual planar graph and partitioning still hold. In addition, having extra edges and regions only adds additional constraints, so it could only help with sufficiency. There exist numerous algorithms for finding the Delaunay Triangulation in $\mathcal{O}(n \log n)$ time (e.g. see \cite{VoronoiKarger}), which we have shown is a valid partitioning. Thus, this concludes our algorithm for finding a valid region partitioning in $\mathcal{O}(n \log n)$ time.
\end{proof}

\end{appendix}

\end{document}